\documentclass[11pt]{article}
\usepackage{times}
\usepackage[tbtags]{amsmath}
\usepackage{amsthm}
 \theoremstyle{plain}
 \newtheorem{theorem}{Theorem}
 
 \newtheorem{corollary}[theorem]{Corollary}

 \theoremstyle{definition}
 \newtheorem{definition}[theorem]{Definition}

 \theoremstyle{remark}

 \theoremstyle{plain}
 \newtheorem*{theorem*}{Theorem}
 \newtheorem*{lemma*}{Lemma}
 \newtheorem*{corollary*}{Corollary}
 \newtheorem*{proposition*}{Proposition}
 \newtheorem*{claim*}{Claim}
 
\usepackage{amssymb}
\usepackage{amsfonts}

\newlength{\actualtopmargin}
\newlength{\actualsidemargin}
\newcommand{\ket}[1]{| #1 \rangle}
\newcommand{\abs}[1]{\vert #1 \vert}
\newcommand{\ceil}[1]{\lceil #1 \rceil}

\newcommand{\bound}[2]{\in\{ #1,\ldots,#2\}}

\newcommand{\Vertice}{V}
\newcommand{\Edge}{E}

\newcommand{\myvec}[1]{\vec #1}

\newcommand{\calH}{\mathcal{H}}

\newcommand{\sfA}{\mathsf{A}}
\newcommand{\sfB}{\mathsf{B}}

\newcommand{\sfQ}{\mathsf{Q}}
\newcommand{\sfR}{\mathsf{R}}
\newcommand{\sfS}{\mathsf{S}}
\newcommand{\sfT}{\mathsf{T}}

\newcommand{\bbC}{\mathbb{C}}

\newcommand{\bbF}{\mathbb{F}}

\newcommand{\Complex}{\bbC}

\newcommand{\bmzero}{\boldsymbol{0}}
\newcommand{\bmone}{\boldsymbol{1}}

\newcommand{\unitaryU}{\mathit{U}}\newcommand{\unitaryZ}{\mathit{Z}}
\newcommand{\unitaryW}{\mathit{W}}
\newcommand{\unitaryY}{\mathit{Y}}

\newcommand{\ignore}[1]{}

\begin{document}



\ignore{
\title{\Large
 \textbf{Constructing Quantum Network Coding Schemes from Classical Nonlinear Protocols
 }\\
}

\author{
 Hirotada Kobayashi\footnotemark[1]\\
 \and
 Fran\c{c}ois Le Gall\footnotemark[2]\\
 \and
 Harumichi Nishimura\footnotemark[3]\\
 \and
 Martin R\"otteler\footnotemark[4]\\
}

\date{}
}



\renewcommand{\thefootnote}{\fnsymbol{footnote}}

\ignore{
\begin{center}
{\large
 \footnotemark[1]%
 National Institute of Informatics, Tokyo, Japan\\
 [1.5mm]
 \footnotemark[2]%
 The University of Tokyo, Tokyo, Japan\\
 [1.5mm]
 \footnotemark[3]%
 Osaka Prefecture University, Sakai, Osaka, Japan\\
 [1.5mm]
 \footnotemark[4]%
 NEC Laboratories America, Inc., Princeton, NJ, USA
}\\
[5mm]
\end{center}
}

\pagestyle{plain}
\begin{center}
{\Large {\bf Constructing Quantum Network Coding Schemes from\\

\medskip

Classical Nonlinear Protocols}}
\\

\

{\sc Hirotada Kobayashi}$^1$ \hspace{5mm}
{\sc Francois Le Gall}$^2$ \hspace{5mm} 
{\sc Harumichi Nishimura}$^3$ \hspace{5mm} 
{\sc Martin R\"otteler}$^4$ 

\

{$^1${National Institute of Informatics, Tokyo, Japan}; {\tt hirotada@nii.ac.jp}

$^2${The University of Tokyo, Tokyo, Japan}; {\tt legall@is.s.u-tokyo.ac.jp}

$^3${Osaka Prefecture University, Japan}; {\tt hnishimura@mi.s.osakafu-u.ac.jp}

$^4${NEC Laboratories America, Inc., Princeton, NJ, USA}; {\tt mroetteler@nec-labs.com} 
}
\end{center}

\

{\bf Abstract.}  The $k$-pair problem in network coding theory asks to send $k$
  messages simultaneously between $k$ source-target pairs over a
  directed acyclic graph.  In a previous paper [ICALP 2009, Part I, pages 622--633] 
  the present authors
  showed that if a classical $k$-pair problem is solvable by means of
  a linear coding scheme, then the quantum $k$-pair problem over the
  same graph is also solvable, provided that classical communication
  can be sent for free between any pair of nodes of the graph. Here we
  address the main case that remained open in our previous work, namely
  whether {\em nonlinear} classical network coding schemes can also give
  rise to quantum network coding schemes. This question is motivated
  by the fact that there are networks for which there are no linear
  solutions to the $k$-pair problem, whereas nonlinear solutions
  exist. In the present paper we overcome the limitation to linear
  protocols and describe a new communication protocol for perfect
  quantum network coding that improves over the previous one as follows:
  (i) the new protocol does not put any condition on the underlying
  classical coding scheme, that is, it can simulate nonlinear
  communication protocols as well, and (ii) the amount of classical
  communication sent in the protocol is significantly reduced.

\

\section{Introduction}   

The idea of {\em network coding}, proposed in the seminal paper by Ahlswede, Cai, Li and Yeung \cite{ACLY:2000}, 
opened up a new communication-efficient way of sending information through networks. 
The key idea is to allow coding and replication of information locally at any intermediate node 
of the network. For instance, this allows one to send two bits simultaneously 
between two source-target pairs over several networks for which the same task cannot be solved
by routing. A simple, yet instructive, example is the butterfly network described in 
Fig.~\ref{fig:classicalbutterfly}. 
We refer to Refs.~\cite{FS:2007,HL:2008,Y:2007,YLCZ:2006} for extensive treatments 
of the topic of classical network coding.

In quantum information processing, it is often desirable to manipulate quantum states 
by applying local operations only, rather than applying global operations 
that require to send quantum information between different places. 
This in particular applies to the situation of communication tasks involving 
quantum information where it is quite natural to assume that whenever quantum information 
is sent over a channel, there is a high chance that it will be corrupted,
whereas classical information can be sent very reliably.
In this context a natural question is whether the concept of network coding can be
applied to quantum networks in order to reduce the amount of
{\em quantum communication}. There have been several studies working on 
``quantum'' network coding \cite{H:2007,HIN+:2007,LOW:2006,SS:2006}. 
These papers deal with the challenge to send quantum information over
a network as well as possible, a task that is greatly hampered 
by the fact due to the no-cloning theorem that unknown quantum information 
cannot be replicated. A natural target problem that has crystallized out from the
prior works \cite{H:2007,HIN+:2007,LOW:2006} as being at the
core of the issue is the following quantum $k$-pair problem: Given a
directed acyclic graph $G$ with $k$ source-target pairs (where we
assume all the edges, which represent quantum channels, have unit
capacities), is there a way of sending $k$ quantum messages between the
$k$ pairs? Note that the classical $k$-pair problem, in which the
channels and messages are classical, is one of the most important
network coding problems (for instance, see Refs.~\cite{DZ:2006,HKL:2006,INPRY:2008,LL:2004,WS:2007}).  
The butterfly network described in Fig.~1 is an instance of the two-pair problem.


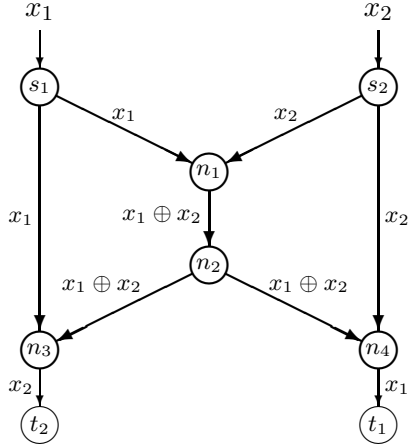
\begin{figure}[t]
\begin{center}
\setlength{\unitlength}{0.25mm}
\begin{picture}(220,230)
\thicklines
\put( 20, 220){\makebox(0,0){$x_1$}}
\put(200, 220){\makebox(0,0){$x_2$}}
\thinlines
\put( 20, 210){\vector(0,-1){20}}
\put(200, 210){\vector(0,-1){20}}
\thicklines
\put( 20, 180){\circle{20}}
\put( 20, 180){\makebox(0,0){\footnotesize $s_1$}}
\put( 20, 170){\vector(0,-1){120}}
\put( 10, 110){\makebox(0,0){\footnotesize $x_1$}}
\put( 28.94, 175.53){\vector(2,-1){72.12}}
\put( 65, 165){\makebox(0,0){\footnotesize $x_1$}}
\put(200, 180){\circle{20}}
\put(200, 180){\makebox(0,0){\footnotesize $s_2$}}
\put(200, 170){\vector(0,-1){120}}
\put(210, 110){\makebox(0,0){\footnotesize $x_2$}}
\put(191.06, 175.53){\vector(-2,-1){72.12}}
\put(151, 165){\makebox(0,0){\footnotesize $x_2$}}
\put(110, 135){\circle{20}}
\put(110, 135){\makebox(0,0){\footnotesize $n_1$}}
\put(110, 125){\vector(0,-1){30}}
\put( 85, 112){\makebox(0,0){\footnotesize ${x_1 \oplus x_2}$}}
\put(110,  85){\circle{20}}
\put(110,  85){\makebox(0,0){\footnotesize $n_2$}}
\put(101.06,  80.53){\vector(-2,-1){72.12}}
\put( 53,  75){\makebox(0,0){\footnotesize ${x_1 \oplus x_2}$}}
\put(118.94,  80.53){\vector(2,-1){72.12}}
\put(163,  75){\makebox(0,0){\footnotesize ${x_1 \oplus x_2}$}}
\put( 20,  40){\circle{20}}
\put( 20,  40){\makebox(0,0){\footnotesize $n_3$}}
\put(200,  40){\circle{20}}
\put(200,  40){\makebox(0,0){\footnotesize $n_4$}}
\thinlines
\put( 200, 00){\makebox(0,0){\footnotesize $t_1$}}
\put( 20, 00){\circle{20}}
\put( 20, 00){\makebox(0,0){\footnotesize $t_2$}}
\put( 200, 00){\circle{20}}
\put( 20,  30){\vector(-0,-1){20}}
\put(  10,  20){\makebox(0,0){\footnotesize $x_2$}}
\put(200,  30){\vector( 0,-1){20}}
\put(210,  20){\makebox(0,0){\footnotesize $x_1$}}
\end{picture}
\caption{
  The butterfly network and a classical linear coding protocol.
  The node $s_1$ (resp.~$s_2$) has for input a bit $x_1$ (resp.~$x_2$).
  The task is to send $x_1$ to $t_1$ and $x_2$ to $t_2$. 
  The capacity of each edge is assumed to be one bit. 
   \label{fig:classicalbutterfly}
}
\end{center}
\end{figure}

Unfortunately, in the early stage of studying quantum network coding
it was shown that there exist networks for which the classical $k$-pair problem is solvable
but the quantum $k$-pair problem is not perfectly solvable \cite{H:2007,HIN+:2007,LOW:2006}. 
For instance, two quantum states cannot be sent simultaneously and perfectly (i.e., with fidelity one) 
between the two source-target pairs in the butterfly network. 
However, the situation changes dramatically if classical communication is allowed freely 
(which seems to be reasonable since classical communication is much cheaper than quantum communication and does not increase the amount of entanglement of the system). 
Indeed, the authors of the present paper established that any linear
classical network coding protocol over ${\mathbb F}_2$ (i.\,e., a
scheme where the encoding operation at each node is a linear function
of its inputs) for the multi-cast problem can be turned into a perfect
quantum network coding protocol \cite{KLN+:2010}. This was generalized
to the $k$-pair case~\cite{KLNR:2009b} where it was shown that if the
classical $k$-pair problem is solvable using a {\em linear} coding
scheme (or even just a vector-linear coding scheme over a finite field
or a finite ring) then the quantum $k$-pair problem is also solvable
using free classical communication.  This result gives rise to two
natural questions.

The first question is whether the linearity condition on the coding schemes of the classical $k$-pair problem 
can be removed. Indeed, there exist classical $k$-pair problems 
that are solvable with nonlinear coding schemes, 
but cannot be solved with linear coding schemes~\cite{DFZ:2005,R:2004}.
This question is closely related to the following open problem: 
can we construct an instance of the $k$-pair problem for which there is a 
(nonlinear) solution to the classical version of the problem, but for which no perfect solution to the 
corresponding quantum version exists, even with free classical communication?
Note that the techniques used in Ref.~\cite{KLNR:2009b} rely on the linearity 
of the classical encoding scheme, and hence they cannot be used directly 
when simulating classical nonlinear coding schemes.

The second question is how much amount of classical communication is sufficient. 
The protocol in Ref.~\cite{KLNR:2009b} essentially uses the fact 
that classical information can be sent (for free) between any two nodes, i.e., 
there exists a classical two-way channel between any two nodes, and their capacities are unlimited. 
Obviously, it would be desirable to find a weaker requirement on the classical communication, 
and to reduce the amount of classical communication as much as possible.
This second question is closely related to the work 
of Leung, Oppenheim and Winter \cite{LOW:2006}. 
They investigated various settings of quantum network coding assisted with 
supplied resources such as free classical communication or entanglement. 
Among others, they considered the case where classical communication can be sent 
only between each pair of nodes connected by a quantum channel and only in the direction of that quantum channel.
 Unfortunately again, they found that the quantum two-pair problem on the butterfly network cannot be solved 
even under this model. One open problem is thus to clarify which types of assistance of 
classical communication enable us to construct a quantum network coding protocol 
for a given $k$-pair problem, and show the minimal amount of classical communication 
necessary under such a model. 

{\bf Our contribution.} This paper provides solutions to both of the above two questions.
 We present a quantum protocol solving, if there is some help of 
classical communication, any instance of the $k$-pair problem for which 
 the corresponding classical version is solvable (under {\em any} coding scheme). 
 In other words, our result shows that whenever an instance of the classical $k$-pair problem is
 solvable, the quantum version of the same problem is solvable when assisted with 
classical communication. Furthermore, classical communication is only sent between two nodes linked 
by quantum channels, and more precisely one unit of classical communication
is sent in the direction of each quantum channel, and one unit is sent in the reverse direction of each 
quantum channel. 
When considering two-dimensional quantum states (qubits), 
each classical communication unit consists of one bit, and thus, at most two bits are sent between adjacent 
nodes: one in the direction of the quantum channel and the other in the reverse direction. 
The total amount of classical communication bits sent is then at most twice the number of edges in the graph.
This significantly improves the bound given in Ref.~\cite{KLNR:2009b}, in which 
the amount of classical communication going through every edge could depend on the number of nodes. 

The starting point of our protocol is the method proposed in Ref.~\cite{KLNR:2009b}.
We first simulate a classical protocol 
by applying a quantum operator at each node in order to simulate the classical encoding performed at this node. 
This simulation introduces intermediate registers that are entangled with the quantum state 
we want to send to the targets, and thus have to be ``properly disentangled.'' 
All the difficulties come from this latter crucial part. The technique used in  
Ref.~\cite{KLNR:2009b} was to measure these intermediate registers in the Fourier basis, 
and then to send the measurement outcomes to the target nodes, who then correct locally 
the phase introduced by the measurements. 
However, this technique relies on the fact that the classical protocol 
being simulated is linear, namely that the exponent in the phase introduced 
is a linear function of the input --- this is why the phase could be corrected locally at the targets. 
In our new protocol, we consider a different way of successfully disentangling 
the intermediate registers. 
The registers are again measured in the Fourier basis, but the measurement outcomes are then sent
to the nodes to which the current node has incoming edges (instead of to the target nodes). 
We then show that, when these operations are done in a proper order (a reverse topological order of the nodes), 
then the phase introduced by the measurements can be corrected locally at these nodes. 
Repeating this process for each internal node of the graph enables us to disentangle 
almost all the intermediate registers. 
The remaining intermediate registers are those owned by the $k$ source nodes, 
which can be disentangled by measuring them in the Fourier basis, 
but now sending the measurement outcomes through the graph to the targets. 
The point is that this can be done in a very communication-efficient way 
since this becomes precisely an instance of the classical $k$-pair problem for which a solution is available. 

In our new protocol, the classical coding scheme we simulate then appears three times. First, this scheme is simulated quantumly, which introduces the intermediate registers --- this uses one unit of quantum communication for each edge (in the original direction of the edge). Second, it is used when removing 
the internal intermediate registers to correct the phase --- this uses one unit of classical communication for each edge (in the reverse direction of the edge). Third, it is used explicitly in order to remove, at the last part of the protocol, the intermediate registers owned by the source nodes --- this uses one unit of classical communication 
for each edge (in the original direction of the edge).

Actually, our techniques can also be used to create EPR-pairs between the sources and the targets of an 
instance of the $k$-pair problem, whenever the associated classical $k$-pair problem is solvable, using one qubit of quantum communication and only one bit of classical communication per edge,
as will be discussed in Section \ref{section:EPR}.
Note that once EPR-pairs are shared, the quantum $k$-pair problem can be solved using teleportation.
However, this would require three bits per edge in total, while the protocol described above (designed specifically for the $k$-pair problem) uses only two bits per edge.

{\bf Related work.}
The restriction of classical communication in this paper was also
studied before by Leung, Oppenheim and Winter \cite{LOW:2006} (who called this 
the two-way assisted case with free two-way classical communication). 
As mentioned in Ref.~\cite{LOW:2006}, this enables
us to send a qubit in a reverse direction to the quantum channel
by sending two bits (via quantum teleportation). Hence some examples
of the $k$-pair problem are shown to be solvable by routing (i.e., without using any coding) when 
{\em time sharing} is allowed
(i.e., when the rate of transmission is studied under the assumption that the network can be used more than once):
see Ref.~\cite{LOW:2006} for the butterfly network. 
It is open, however, whether every instance is solvable by routing. 
Our method suggests a different way of solving the quantum $k$-pair problem, 
which requires some coding but does not require time sharing, and works for 
any solvable classical instance.

\section{\boldmath{The $k$-pair problem}}\label{section:statement}
\subsubsection*{\bf The classical \boldmath{$k$}-pair problem.}
We recall the statement of the $k$-pair problem in the classical case 
(often called the multiple unicast problem), and the definition of a solution to this problem.
The reader is referred to, for example, Ref.~\cite{DZ:2006} for further details. 
We use the same setting as in Ref.~\cite{KLNR:2009b}, 
but some conventions differ in order to facilitate the exposition of our new protocol.

An instance of a $k$-pair problem is a directed acyclic graph ${G=(\Vertice,\Edge)}$ and $k$ pairs $(s_1,t_1),\ldots,(s_k,t_k)$
of nodes in $V$. 
We suppose, without loss of generality, that the nodes $s_1,\ldots, s_k$ have fan-in zero, and the nodes $t_1,\ldots, t_k$ have fan-in one and fan-out zero.
We denote by $\overline{\Edge}$ the set
of {\em internal edges}, i.e., the edges in $\Edge$ with no extremity in $\{t_1,\ldots,t_k\}$. The $k$ edges with an 
extremity in $\{t_1,\ldots,t_k\}$ are called the {\em output edges}. For each node~$v$ in~$\Vertice$, we fix an arbitrary ordering 
of the incoming edges of $v$, and an arbitrary ordering of the outgoing edges of $v$.

For each $i\bound{1}{k}$, a message $x_i$ is given at the source~$s_i$, and has to  
be sent to the target~$t_i$ through $G$ under the condition 
that each edge has unit capacity. For convenience, the following convention is assumed when describing a classical coding scheme. 
Each source ${s_i}$ is supposed to have a ``virtual'' incoming edge from which it receives its input $x_i$ (each source node has 
thus fan-in at least one, but these virtual edges are not included in $\Edge$).
In this way, the source nodes perform coding operations on their inputs,
and this convention enables one to ignore the distinction between
source nodes and internal nodes.
These conventions are illustrated in Fig.~\ref{fig:classicalbutterfly}.

Let $\Sigma$ be a finite set.
A {\em coding scheme} over $\Sigma$ is a choice of operations for each node in $V$ 
with nonzero fan-in and nonzero fan-out: 
for each node $v\in\Vertice$ with fan-in~$m\ge 1$ and fan-out~$n\ge 1$, the operation at $v$ 
is written as $n$ functions $f_{v,1},\ldots,f_{v,n}$, each from $\Sigma^m$ 
to $\Sigma$, where the value $f_{v,i}(z_1,\ldots,z_m)$ represents the message 
sent through the $i$-th outgoing edge of $v$ when the inputs of the $m$ incoming edges are
$z_1,\ldots,z_m$.  Since the graph $G$ is acyclic,  we can fix a topological ordering of the nodes of the graphs,
i.e., an ordering in which each node comes before all nodes to which it has outgoing edges. The coding scheme can then be explicitly implemented by 
applying the encoding functions in the order in which the nodes appear.
A {\em solution} over $\Sigma$ to an instance of the $k$-pair problem is a coding scheme over $\Sigma$ 
that enables one to send simultaneously $k$ messages $x_i\in\Sigma$ from $s_i$ to $t_i$, for all~$i\bound{1}{k}$.
For example, the coding scheme in Fig.~\ref{fig:classicalbutterfly} is a solution over 
$\{0,1\}$
to the two-pair problem associated with the butterfly graph.

\subsubsection*{\bf The quantum \boldmath{$k$}-pair problem.}
We suppose that the reader is familiar with the basics of quantum information theory and 
refer to Ref.~\cite{NC:2000} for a good reference.
In this paper we use the same model for the quantum $k$-pair problem as in Ref.~\cite{KLNR:2009b} 
except restricting the classical communication to be allowed.

An instance of a quantum $k$-pair problem is, as in the classical case, a directed acyclic graph $G=(\Vertice,\Edge)$ and $k$ pairs of nodes $(s_1,t_1),\ldots,(s_k,t_k)$. 
Let $\calH$ be a Hilbert space. A (quantum) solution for $\calH$ to the instance  
is a quantum coding scheme, i.e., a choice of quantum operations over all nodes, 
allowing us to send a quantum state $\ket{\psi}\in\calH^{\otimes k}$ supported on the source nodes $s_1,\ldots,s_k$ 
(in this order) to the target nodes $t_1,\ldots,t_k$ (in this order). 
We consider the model where each edge of  $G$ can transmit one quantum state over~$\calH$. 
In this paper, classical communication is only allowed 
between any two adjacent nodes: if $(v_1,v_2)\in E$ then classical communication is possible
from $v_1$ to $v_2$ and from $v_2$ to $v_1$.
(Note that in Ref.~\cite{KLNR:2009b} classical communication was allowed between any two nodes of $G$.)
For a positive integer~$d$, an instance of the quantum $k$-pair problem is said to be \emph{solvable} 
over $\Complex^d$ if there exists a protocol solving this problem for~$\calH=\Complex^d$.

\section{Main result}\label{section:main}

The main result of this paper is the following theorem. 
\begin{theorem}\label{theorem1}
Let $G=(\Vertice,\Edge)$ be a directed acyclic graph and $(s_1,t_1),\ldots,(s_k,t_k)$ be $k$ pairs of nodes in $V$.
Let $\Sigma$ be a finite set.
Suppose that there exists a solution over $\Sigma$ to the associated classical $k$-pair problem.
Then the corresponding quantum $k$-pair problem is solvable over $\Complex^{|\Sigma|}$. 
Moreover, there exists a quantum protocol for this task that sends at most  two elements of $\Sigma$ per edge as classical communication (one in each direction of the edge), 
i.e., at most $2\abs{\Edge}\ceil{\log_2 \abs{\Sigma}}$ bits of classical communication in total.
\end{theorem}


The amount of free classical communication used in the protocol of our previous work~\cite{KLNR:2009b} was 
$km\abs{V}\ceil{\log_2 \abs{\Sigma}}$ bits, where $m$ denotes the maximal fan-in over all nodes in $G$ 
(note that $|E|$ is at most $O(m|V|)$ in the worst case, and is much smaller than $m|V|$ in most cases). 
The bound in Ref.~\cite{KLNR:2009b} actually supposes that classical information can be sent 
between any two nodes of the graph. If the classical communication is restricted to adjacent nodes, 
then the amount of communication can increase by a factor
corresponding to the length of the longest path from an internal node to a target node, and, in general, the number of elements of $\Sigma$ that are sent through a given edge depends on $k$, $\abs{V}$ and $m$.
In comparison, Theorem~\ref{theorem1} enables us to perform quantum network coding 
by sending at most two elements of $\Sigma$ between adjacent nodes.    

Notice that the classical communication bound of Theorem \ref{theorem1} 
is described independently of the original classical coding scheme. In fact,
if there exists a coding scheme on a subgraph $(V,\Edge')$ of $G$, where $\Edge'\subset \Edge$
(e.g., if a part of the graph is not involved in the original classical encoding scheme), 
then the factor $\abs{\Edge}$ can be improved in a straightforward way.  
We can then also describe the complexity of a quantum $k$-pair problem in term of the complexity 
of the best classical protocol for the corresponding classical task (i.e., in term of the smallest subset 
$\Edge'\subseteq \Edge$ such that a solution over $(V,\Edge')$ exists).
A specific statement of this observation for the case $|\Sigma|=2$ follows.
\begin{corollary}
Let $G=(\Vertice,\Edge)$ be a directed acyclic graph and $(s_1,t_1),\ldots,(s_k,t_k)$ be $k$ pairs of nodes in $V$.
Suppose that there exists a solution over $\{0,1\}$ to the associated classical $k$-pair problem using a total amount of $C$ bits of communication.
Then there exists a quantum protocol over $\Complex^{2}$ for the corresponding quantum $k$-pair problem that sends in total $C$ qubits of communication and 
$2C$ bits of classical communication.
\end{corollary}

\section{Protocol}\label{sec:protocol}

As in Ref.~\cite{KLNR:2009b}, the basic strategy for proving Theorem \ref{theorem1} is 
to perform a quantum simulation of the classical coding scheme, while the simulation method is more elaborated.
Before presenting the proof, we need some preliminaries.

\subsection{Quantum operators}
Let $\Sigma$ be a finite set.
In the quantum setting, we suppose that each register contains a quantum state 
over $\calH=\Complex^{|\Sigma|}$, and denote by $\{\ket{z}\}_{z\in \Sigma}$ an orthonormal basis of $\calH$.
We use the notation $\ket{0_{\calH}}$ to refer to an arbitrary vector of the basis 
that will be used to initialize registers.
Let $\sigma$ be an arbitrary bijection from $\Sigma$ to the set of integers $\{0,\ldots,\abs{\Sigma}-1\}$. 
For convenience we denote $\overline x=\sigma(x)$ for any element $x\in \Sigma$. 

We define a unitary operator $\unitaryW$ over the Hilbert space $\calH$ as follows: 
for any $y\in \Sigma$, the operator $\unitaryW$ maps the basis state $\ket{y}$ to the state
$$
\frac{1}{\sqrt {|\Sigma|}}\sum_{z\in \Sigma}\exp\Bigl(\frac{2\pi \iota \overline{y}\cdot \overline{z}}{\abs{\Sigma}}\Bigr)\ket{z}.
$$
Note that $\unitaryW$ is basically the quantum Fourier transform over the cyclic group of order $\abs{\Sigma}$. 

A measurement of a quantum state over the Hilbert space $\calH$ in the Fourier basis 
consists in applying the operator $\unitaryW$ to the quantum state, 
and then measuring it in the basis $\{\ket{z}\}_{z\in \Sigma}$. 
The measurement outcome is an element of $\Sigma$.
Notice that, if the quantum state measured in the Fourier basis is $\ket{y}$, for some $y\in\Sigma$, 
and the measurement outcome is $z\in\Sigma$, then the state after measurement
becomes $\exp\bigl(\frac{2\pi \iota \overline{y}\cdot \overline{z}}{\abs{\Sigma}}\bigr)\ket{z}.$

Let $m$ and $n$ be two positive integers and $f_1,\ldots,f_n$ be $n$ functions from $\Sigma^m$ to $\Sigma$.
Let $\unitaryU_{f_1,\ldots,f_n}$ be a unitary operator over the Hilbert space ${\calH^{\otimes m}\otimes\calH^{\otimes n}}$ 
such that, for any elements $y_1,\ldots,y_m$ in $\Sigma$, the operator
$\unitaryU_{f_1,\ldots,f_n}$ maps the basis state $\ket{y_1,\ldots,y_m}\ket{0_{\calH},\ldots,0_{\calH}}$ to the state
$\ket{y_1,\ldots,y_m}\ket{f_{1}(y_1,\ldots,y_m),\ldots,f_{n}(y_1,\ldots,y_m)}$. 

\subsection{Global encoding functions}
A coding scheme over a directed acyclic graph $G=(\Vertice,\Edge)$ 
naturally induces a {\em global encoding function} associated with each edge in $\Edge$.
Let $\myvec{e}$ be an edge in $\Edge$. Then the global encoding function $g_{\myvec{e}}\colon \Sigma^k\to \Sigma$ 
associated with the edge $\myvec{e}$ is the function of the variables $x_1,\ldots, x_k$ representing the message sent through this
edge by the coding scheme when the input is $x_1,\ldots, x_k$. Since $G$ is acyclic, the global encoding functions can be defined directly as follows.
\begin{definition}\label{def:global}
Let $G=(\Vertice,\Edge)$ be a directed acyclic graph and $(s_i,t_i)$ be $k$ pairs of nodes in $\Vertice$, 
for $i\bound{1}{k}$. 
Consider an encoding scheme for this $k$-pair problem defined by, for each node $v\in\Vertice$ with fan-in $m$ and fan-out $n$, functions $f_{v,1},\ldots,f_{v,n}$ from $\Sigma^m$ to $\Sigma$ 
representing the encoding performed at $v$. Then the global encoding function $g_{\myvec{e}}\colon \Sigma^k\to \Sigma$ associated with this encoding scheme, for each $\myvec{e}\in\Edge$, is inductively defined as follows.
\begin{itemize}
\item[(a)]
Suppose that $\myvec{e}$ is an edge in $\Edge$ of the form $(s_i,u)$, with $i\bound{1}{k}$ and $u\in \Vertice$. Suppose that $\myvec{e}$ is the $\ell$-th outgoing edge of $s_i$.  
Then $g_{\myvec{e}}(x_1,\ldots,x_k)=f_{s_i,\ell}(x_i)$ for any $x_1,\ldots,x_k\in\Sigma$.
\item[(b)]
Suppose that $\myvec{e}$ is an edge in $\Edge$ of the form $(v,w)$ with $v\notin\{s_1,\ldots,s_k\}$. Suppose  that $\myvec{e}$ is the $\ell$-th outgoing edge of $v$
and denote by $(u_1,v),\ldots,(u_m,v)$ the incoming edges of node $v$ (see Fig.~\ref{fig:globalfunctions}). 
Then, for any $x_1,\ldots,x_k\in\Sigma$,
\begin{equation}\label{eq:global-local}
g_{(v,w)}(x_1,\ldots,x_k)=f_{v,\ell}(g_{(u_1,v)}(x_1,\ldots,x_k),\ldots,g_{(u_m,v)}(x_1,\ldots,x_k)).
\end{equation}
\end{itemize}
\end{definition}
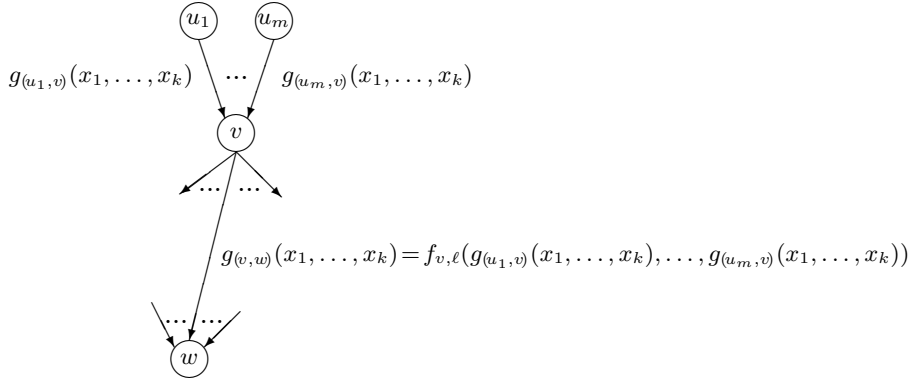
\begin{figure}[t]
\begin{center}
\setlength{\unitlength}{0.25mm}
\begin{picture}(140,200)(100,40)
\put( 25, 40){\circle{20}}
\put( 25, 40){\makebox(0,0){\footnotesize $w$}}
\put( 50, 160){\circle{20}}
\put( 50, 160){\makebox(0,0){\footnotesize $v$}}
\put( 30, 220){\makebox(0,0){\footnotesize $u_1$}}
\put( 70, 220){\makebox(0,0){\footnotesize $u_m$}}
\put( 30, 220){\circle{20}}
\put( 70, 220){\circle{20}}
\put( 50, 150){\vector(-1,-4){25}}
\put( 36, 130){\makebox(0,0){$...$}}
\put( 57, 130){\makebox(0,0){$...$}}
\put( 37, 60){\makebox(0,0){$...$}}
\put( 50, 190){\makebox(0,0){$...$}}
\put( 18, 60){\makebox(0,0){$...$}}

\put(-22, 190){\makebox(0,0){\footnotesize $g_{(\!u_1,v\!)}(x_1,\ldots,x_k)$}}
\put( 125, 190){\makebox(0,0){\footnotesize $g_{(\!u_m,v\!)}(x_1,\ldots,x_k)$}}
\put( 225, 95){\makebox(0,0){\footnotesize $g_{(\!v,w\!)}(x_1,\ldots,x_k)\!=\!f_{v,\ell}(g_{(\!u_1,v\!)}(x_1,\ldots,x_k),\ldots,g_{(\!u_m,v\!)}(x_1,\ldots,x_k))$}}

\put( 50, 150){\vector(1,-1){24}}
\put( 50, 150){\vector(-4,-3){30}}
\put( 52, 65){\vector(-1,-1){19}}
\put( 30, 210){\vector(1,-3){14}}
\put( 70, 210){\vector(-1,-3){14}}
\put( 5, 70){\vector(1,-2){12}}

\end{picture}
\caption{
  Illustration of the notion of global encoding functions. Here node $v$ has $m$ incoming edges labelled in increasing order $(u_1,v),\cdots,(u_m,v)$,
  and edge $(v,w)$ is the $\ell$-th outgoing edge of $v$.
   \label{fig:globalfunctions}
}
\end{center}
\end{figure}

\subsection{Proof of Theorem \ref{theorem1}}
Now we are ready to give the proof of Theorem \ref{theorem1} 
(see Appendix for an illustration of our strategy on the butterfly graph).

\begin{proof}[Proof of Theorem \ref{theorem1}]
Let $G=(\Vertice,\Edge)$ be a graph on which there exists a solution
to the classical $k$-pair problem associated with the pairs $(s_i,t_i)$. 
Without loss of generality, we suppose that all nodes in $V\backslash\{t_1,\ldots,t_k \}$
have nonzero fan-in and nonzero fan-out (remember that each source node $s_i$ has one virtual incoming edge).
For each node $V\backslash\{t_1,\ldots,t_k \}$ 
with fan-in $m\ge 1$ and fan-out $n\ge 1$, let $f_{v,1},\ldots,f_{v,n}$ be the coding operations performed at node $v$
in the solution, where each function $f_{v,i}$ is from $\Sigma^m$ to $\Sigma$.
 
Suppose that the input state of the quantum task is
$$
\ket{\psi_S}_{(\sfS_1,\ldots,\sfS_k)}=\sum_{x_1,\ldots,x_k\in \Sigma}\alpha_{x_1,\ldots,x_k}\ket{x_1}_{\sfS_1}\otimes\cdots\otimes\ket{x_k}_{\sfS_k},
$$
where the $\alpha_{x_1,\ldots,x_k}$'s are complex numbers such that $\sum_{x_1,\ldots,x_k\in \Sigma}\abs{\alpha_{x_1,\ldots,x_k}}^2=1$.
Here, for each $i\bound{1}{k}$, $\sfS_i$ is a register received by the node $s_i$ from its incoming virtual edge.
Our protocol consists of the three parts that are described in details below.\vspace{2mm}

\noindent{\bf Part 1.}
First, we simulate the solution to the associated classical task node by node. 
In our simulation, each node in $\Vertice$ will receive one quantum register along each incoming edge 
(remember that each source $s_i$ receives $\sfS_i$ from its incoming virtual edge).
Each node in $V\backslash\{t_1,\ldots,t_k\}$ will perform,  in the same order as in the classical protocol (e.g., any topological order of the nodes of $G$), a quantum operation
on the registers it receives, and send new registers through its outgoing edges. 

More precisely, let $v\in\Vertice\backslash\{t_1,\ldots,t_k\}$ be a node of $G$ with fan-in $m$ and fan-out $n$. 
Let $\sfQ_1,\ldots,\sfQ_m$ denote the quantum registers received by the incoming edges.
The coding performed at node $v$ is simulated as follows:  $n$ quantum registers
$\sfQ'_1,\ldots,\sfQ'_n$, each initialized to $\ket{0_{\calH}}$, are introduced, and then
the operator $\unitaryU_{f_1,\ldots,f_n}$ is applied to the registers $(\sfQ_1,\ldots,\sfQ_m,\sfQ'_1,\ldots,\sfQ'_n)$.
The registers $\sfQ'_1,\ldots,\sfQ'_n$ are then sent along the $n$ outgoing edges of $v$ (i.e., $\sfQ'_i$ is sent along the $i$-th outgoing edge of $v$), and the registers $\sfQ_1,\ldots,\sfQ_m$ 
are kept at node $v$.

Such a simulation is done for all the nodes in $\Vertice\backslash\{t_1,\ldots,t_k\}$. Notice that exactly one register is introduced per edge in $\Edge$. These registers can then be indexed by the edges in $\Edge$: 
the register associated with an edge of the form $(u,v)\in\Edge$ is created at node $u$, then sent from $u$ to $v$ through edge $(u,v)$ in the procedure described above, and is finally owned by node $v$. We will denote 
the registers associated with the output edges (and owned by the target nodes) by $\sfT_1,\ldots, \sfT_k$ and the other registers by $\sfR_{\myvec{e}}$ for each $\myvec{e}\in\overline{E}$. 
Remember that, additionally, each source $s_i$ owns the input register $\sfS_i$, for $i\bound{1}{k}$. Then the overall quantum state consists of $\abs{\Edge}+k$ registers and, since the coding scheme under consideration solves the classical task, 
 is of the form 
\begin{equation*}\label{eq:state_part1}
\begin{split}
\sum_{x_1,\ldots,x_k\in \Sigma}\alpha_{x_1,\ldots,x_k}\ket{x_1}_{\sfS_1}\ket{x_1}_{\sfT_1}\otimes\cdots\otimes\ket{x_k}_{\sfS_k}\ket{x_k}_{\sfT_k}\otimes
\Big(\bigotimes_{\myvec{e}\in\overline{\Edge}}\ket{g_{\myvec{e}}(x_1,\ldots,x_k)}_{\sfR_{\myvec{e}}}\Big),
\end{split}
\end{equation*}
where $g_{\myvec{e}}$'s are the global encoding functions associated with the considered coding scheme, as defined in Definition \ref{def:global}. \vspace{2mm}

\noindent{\bf Part 2.}
Second, we remove all the intermediate registers~$\sfR_{\myvec{e}}$.
Consider an edge $(v,w)\in\overline{\Edge}$ and the register $\sfR_{(v,w)}$ associated with it. This register is owned by the internal node $w$ at the end of the procedure described in Part 1.
Node $w$ first measures register  $\sfR_{(v,w)}$ in the Fourier basis.
Suppose that the outcome of the measurement is $a_{(v,w)}\in\Sigma$. The quantum state then becomes
\[
\begin{split}
&
\sum_{x_1,\ldots,x_k\in \Sigma}\alpha_{x_1,\ldots,x_k}\exp\Big(2\pi\iota\frac{\overline{a_{(v,w)}}\cdot \overline{g_{(v,w)}(x_1,\ldots,x_k)}}{\abs{\Sigma}}\Big)\times
\ket{x_1}_{\sfS_1}\ket{x_1}_{\sfT_1}\otimes\cdots\\
&\hspace{10mm}\cdots\otimes\ket{x_k}_{\sfS_k}\ket{x_k}_{\sfT_k}\otimes
\ket{a_{(v,w)}}_{\sfR_{(v,w)}}\otimes\Big(\bigotimes_{\myvec{e}\in\overline{\Edge}\backslash \{(v,w)\}}\ket{g_{\myvec{e}}(x_1,\ldots,x_k)}_{\sfR_{\myvec{e}}}
\Big).
\end{split}
\]
Register $\sfR_{(v,w)}$ is not anymore entangled with the other registers, and can then be disregarded.
Node $w$ then sends the element $a_{(v,w)}$ of $\Sigma$ to the node $v$ using classical communication. 
Now suppose that $(v,w)$ is the $\ell$-th outgoing edge of $v$ and denote by $(u_1,v),\ldots,(u_m,v)$ the incoming edges of node $v$ (see Fig.~\ref{fig:globalfunctions}). 
Node $v$ then applies the map $\unitaryY_v$ to its registers~$(\sfR_{(u_1,v)},\ldots,\sfR_{(u_m,v)})$, where $\unitaryY_v$ is defined as
\begin{eqnarray*}
\unitaryY_v\colon\ket{z_1,\ldots,z_m}&\mapsto&
\exp\Big({-2\pi \iota \frac{\overline{a_{(v,w)}}\cdot \overline{f_{v,\ell}(z_1,\ldots,z_m)}}{\abs{\Sigma}}}\Big)\ket{z_1,\ldots,z_m}
\end{eqnarray*}
for any $z_1,\ldots,z_m\in \Sigma$. 
From Equation~(\ref{eq:global-local}), this implies that the quantum state becomes
\[
\sum_{x_1,\ldots,x_k\in \Sigma}\alpha_{x_1,\ldots,x_k}
\ket{x_1}_{\sfS_1}\ket{x_1}_{\sfT_1}\otimes\cdots\otimes\ket{x_k}_{\sfS_k}\ket{x_k}_{\sfT_k}
\otimes\Big(\bigotimes_{
\begin{subarray}{c}
\myvec{e}\in\overline{\Edge}\\
\myvec{e}\neq (v,w)
\end{subarray}}
\ket{g_{\myvec{e}}(x_1,\ldots,x_k)}_{\sfR_{\myvec{e}}}
\Big),
\]
where the register $\sfR_{(v,w)}$ has been disregarded. Notice that only one element of $\Sigma$ has been sent (from $w$ to $v$, i.e., in the reverse direction of the edge~$(v,w)$).

The above procedure is performed by each internal node (i.e., each node in~$\Vertice\backslash\{s_1,\ldots,s_k,$ $t_1,\ldots,t_k\}$) in a reverse topological order
(i.e., an ordering in which each node comes before all nodes from which it has incoming edges)
 --- this latter condition is crucial for the correctness of the technique. 
The final state is
\begin{equation}\label{eq:state_part2}
\sum_{x_1,\ldots,x_k\in \Sigma}\alpha_{x_1,\ldots,x_k}\ket{x_1}_{\sfS_1}\ket{x_1}_{\sfT_1}\otimes\cdots\otimes\ket{x_k}_{\sfS_k}\ket{x_k}_{\sfT_k}.
\end{equation}
The total number of elements of $\Sigma$ sent as classical communication in Part 2 is $\abs{\overline{E}}$. More precisely, one element of $\Sigma$ is sent per edge,
in the reverse direction of the edge. \vspace{2mm}

\noindent{\bf Part 3.}
Finally, we remove the registers $\sfS_1,\ldots,\sfS_k$.
Remember that each register $\sfS_i$ is owned by node $s_i$, for $i\bound{1}{k}$. 
Each node $s_i$ then measures its register $\sfS_i$ in the Fourier basis, obtaining an element $b_i\in\Sigma$. The state then becomes
$$
\sum_{x_1,\ldots,x_k\in \Sigma}\alpha_{x_1,\ldots,x_k}\exp\Big(2\pi\iota\frac{\overline{b_1}\cdot \overline{x_1}+\cdots+\overline{b_k}\cdot \overline{x_k}}{\abs{\Sigma}}\Big)
\ket{b_1}_{\sfS_1}\ket{x_1}_{\sfT_1}\otimes\cdots\otimes\ket{b_k}_{\sfS_k}\ket{x_k}_{\sfT_k}.
$$
The registers $\sfS_1,\ldots,\sfS_k$ can then be disregarded.
Each source $s_i$ now sends $b_i$ to the target~$t_i$ using classical communication. 
Notice that this is precisely an instance of the classical $k$-pair problem we are considering, 
and hence the assumed classical network coding protocol is available. 
Therefore, this can be done by sending at most one element of $\Sigma$ per edge. 
As the last step, each target node $t_i$ for $i\bound{1}{k}$ 
applies the map $\unitaryZ_i$ to its register~$\sfT_i$, where $\unitaryZ_i$ maps
the basis state $\ket{x}$ to the state $\exp\big(-2\pi \iota \frac{\overline{b_i}\cdot \overline{x}}{\abs{\Sigma}}\big)\ket{x}$,
for any $x\in \Sigma$. This corrects the phase and the resulting state is 
$$
\sum_{x_1,\ldots,x_k\in \Sigma}\alpha_{x_1,\ldots,x_k}
\ket{x_1}_{\sfT_1}\otimes\cdots\otimes\ket{x_k}_{\sfT_k}
=
\ket{\psi_S}_{(\sfT_1,\ldots,\sfT_k)}.$$
The amount of classical communication sent in Part 3 is at most one element of $\Sigma$ per edge 
(in the direction of the edge), i.e., at most $\abs{\Edge}$ elements of $\Sigma$ in total.\vspace{2mm}

This concludes the proof of Theorem \ref{theorem1}.
\end{proof}

\section{Sharing EPR-pairs}\label{section:EPR}
In this section we describe how our techniques can be used to create EPR-pairs between the sources and the targets of an 
instance of the $k$-pair problem using one qubit of quantum communication and only one bit of classical communication per edge.

Let $\calH$ be a Hilbert space of dimension $2$ with orthonormal basis $\{\ket{0},\ket{1}\}$. 
Let $\sfA$ and $\sfB$ be two registers over $\calH$.
Remember that an EPR-pair over $(\sfA,\sfB)$, denoted by $\ket{\mathrm{\Psi}}_{(\sfA,\sfB)}$, is the quantum state $\frac{1}{\sqrt{2}}(\ket{0}_{\sfA}\ket{0}_{\sfB}+\ket{1}_{\sfA}\ket{1}_{\sfB})$. 

Let $G=(\Vertice,\Edge)$ be a directed acyclic graph and $(s_1,t_1),\ldots,(s_k,t_k)$ be $k$ pairs of nodes.
The corresponding \emph{EPR-pair sharing problem} asks to create $k$ EPR-pairs
$\ket{\mathrm{\Psi}}_{(\sfS_i \sfT_i)}$ where, for $i\bound{1}{k}$, each $\sfS_i$ is a register owned by the source 
$s_i$ and each $\sfT_i$ is a register owned by the target $t_i$. 
We consider the model where each edge of  $G$ can transmit one qubit and classical communication is only allowed 
between any two adjacent nodes.

Suppose that the classical $k$-pair problem associated to $G$ and the 
$k$ pairs $(s_i,t_i)$ has a solution over $\{0,1\}$. 
Each source node $s_i$ creates a register $\sfS_i$ initialized to the state $\ket{0}_{\sfS_i}$, and applies 
a Hadamard transform on it. The quantum state obtained is 
$$\ket{\psi}_{(\sfS_1,\ldots,\sfS_k)}=\frac{1}{\sqrt{2^{k}}}(\ket{0}_{\sfS_1}+\ket{1}_{\sfS_1})\otimes\cdots\otimes(\ket{0}_{\sfS_k}+\ket{1}_{\sfS_k}).$$
Now consider the quantum $k$-pair problem associated to $G$ and the 
$k$ pairs $(s_i,t_i)$ on the input $\ket{\psi}_{(\sfS_1,\ldots,\sfS_k)}$.
Applying Part 1 and Part 2 of the protocol described in Section \ref{sec:protocol} gives (see Formula (\ref{eq:state_part2})) the state
\[
\frac{1}{\sqrt{2^{k}}}(\ket{0}_{\sfS_1}\ket{0}_{\sfT_1}+\ket{1}_{\sfS_1}\ket{1}_{\sfT_1})\otimes\cdots\otimes(\ket{0}_{\sfS_k}\ket{0}_{\sfT_k}+\ket{1}_{\sfS_k}\ket{1}_{\sfT_k})
=\ket{\mathrm{\Psi}}_{(\sfS_1,\sfT_1)}\otimes\cdots\otimes\ket{\mathrm{\Psi}}_{(\sfS_k,\sfT_k)},
\]
while only one bit of classical communication is sent per edge in $\overline{\Edge}$, in the reverse direction of the edge.
Here $\sfT_i$ denotes a register owned by the target node $t_i$, for $i\bound{1}{k}$.
We thus obtain the following result.
\begin{theorem}\label{theorem2}
Let $G=(\Vertice,\Edge)$ be a directed acyclic graph and $(s_1,t_1),\ldots,(s_k,t_k)$ be $k$ pairs of nodes in $V$.
Suppose that there exists a solution over $\{0,1\}$ to the associated classical $k$-pair problem.
Then there exists a quantum protocol that solves the corresponding EPR-pair sharing problem using only one bit of classical communication
per edge (in the reverse direction of the edge).
\end{theorem}



\section*{Acknowledgments}
The authors are grateful to Kazuo Iwama for helpful comments on this work.

\bibliographystyle{abbrv}
\bibliography{KobLeGNisRotQIP11}

\newpage 
\section*{Appendix: Example for our protocol}

We illustrate in this appendix the techniques developed in this paper 
for the butterfly network shown in Figure~\ref{fig:butterfly}. Our task is to send 
a quantum state $\ket{\psi_S}$ from the source nodes $s_1$ and $s_2$ to the target nodes $t_1$ and $t_2$. 
The task can be achieved
using the quantum protocol given in Theorem~\ref{theorem1}.
We give the explicit details for this example.
More precisely, we describe how the protocol simulates
the classical linear coding scheme over $\Sigma=\bbF_2$ (the finite field of size 2) presented in Fig.~\ref{fig:classicalbutterfly}.
Notice that this example does not require nonlinear coding, but we prefer to explain our protocol on this 
simple example rather than on known instances of the $k$-pair problem that require nonlinear coding 
(e.g., the networks proposed in Refs.~\cite{DFZ:2005,R:2004})
since the latter are rather complex. Moreover, we believe that this choice enables the reader 
to better compare our new protocol to the one proposed in our previous work~\cite{KLNR:2009b}.

Hereafter, all the registers are assumed to be single-qubit registers, 
i.e., $\calH=\Complex^2$. 
We denote by $\{\ket{z}\}_{z\in\bbF_2}$ an orthonormal basis of $\calH$. 
All the registers are supposed to be initialized to $\ket{0}$.

\begin{figure}[h]
\begin{center}
\setlength{\unitlength}{0.25mm}
\begin{picture}(220,230)
\thinlines
\put( 20, 220){\makebox(0,0){$\sfS_1$}}
\put(200, 220){\makebox(0,0){$\sfS_2$}}
\put( 20, 210){\vector(0,-1){20}}
\put(200, 210){\vector(0,-1){20}}
\thicklines
\put( 20, 180){\circle{20}}
\put( 20, 180){\makebox(0,0){\footnotesize $s_1$}}
\put( 20, 170){\vector(0,-1){120}}
\put( 10, 115){\makebox(0,0){$\sfR_1$}}
\put( 28.94, 175.53){\vector(2,-1){72.12}}
\put( 65, 170){\makebox(0,0){$\sfR_2$}}
\put(200, 180){\circle{20}}
\put(200, 180){\makebox(0,0){\footnotesize $s_2$}}
\put(200, 170){\vector(0,-1){120}}
\put(210, 115){\makebox(0,0){$\sfR_3$}}
\put(191.06, 175.53){\vector(-2,-1){72.12}}
\put(155, 170){\makebox(0,0){$\sfR_4$}}
\put(110, 135){\circle{20}}
\put(110, 135){\makebox(0,0){\footnotesize $n_1$}}
\put(110, 125){\vector(0,-1){30}}
\put(100, 115){\makebox(0,0){$\sfR_5$}}
\put(110,  85){\circle{20}}
\put(110,  85){\makebox(0,0){\footnotesize $n_2$}}
\put(101.06,  80.53){\vector(-2,-1){72.12}}
\put( 65,  75){\makebox(0,0){$\sfR_6$}}
\put(118.94,  80.53){\vector(2,-1){72.12}}
\put(155,  75){\makebox(0,0){$\sfR_7$}}
\put( 20,  40){\circle{20}}
\put( 20,  40){\makebox(0,0){\footnotesize $n_3$}}
\put(200,  40){\circle{20}}
\put(200,  40){\makebox(0,0){\footnotesize $n_4$}}
\put( 20,  30){\vector(-0,-1){20}}
\put(  8,  20){\makebox(0,0){$\sfT_2$}}
\put(200,  30){\vector( 0,-1){20}}
\put(212,  20){\makebox(0,0){$\sfT_1$}}
\put( 200, 00){\makebox(0,0){\footnotesize $t_1$}}
\put( 20, 00){\circle{20}}
\put( 20, 00){\makebox(0,0){\footnotesize $t_2$}}
\put( 200, 00){\circle{20}}

\end{picture}
\caption{
Quantum network coding through the butterfly network. 
Each edge has quantum capacity one. The task is to send a given input quantum state~$\ket{\psi_S}$ in ${(\sfS_1, \sfS_2)}$
 to ${(\sfT_1, \sfT_2)}$ in this order of registers.
 Here, the quantum register~$\sfS_1$~(resp.~$\sfS_2$) is possessed by the source node~$s_1$~(resp.~$s_2$),
 while the quantum register~$\sfT_1$~(resp.~$\sfT_2$) is possessed at the end of the protocol by the target node~$t_1$~(resp.~$t_2$).
 The protocol given in Theorem~\ref{theorem1} realizes
 perfect quantum transmission of $\ket{\psi_S}$.
 Each $\sfR_i$ or $\sfT_i$ indicates the quantum register to be sent
 along the corresponding edge in the protocol.
 \label{fig:butterfly}
}
\end{center}
\end{figure}
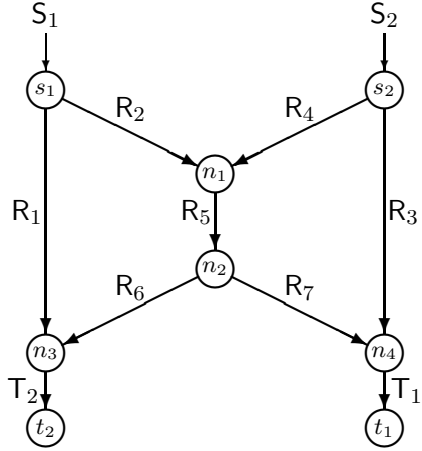

In this example, the unitary operator $\unitaryU_{f_I,f_I}$ is applied at nodes $s_1,s_2$ and $n_2$, 
where $f_I$ denotes the identity map over $\bbF_2$.
The operator $\unitaryU_{f_I,f_I}$ maps each basis state $\ket{y}\ket{z_1,z_2}$ to the state $\ket{y}\ket{z_1+y,z_2+y}$.
The quantum unitary operator $\unitaryU_{f_+}$ is applied at nodes $n_1,n_3$ and $n_4$, 
where $f_+\colon(\bbF_2)^2\to\bbF_2$ is the function mapping $(y_1,y_2)$ to $y_1+y_2$.   
The operator  $\unitaryU_{f_+}$ maps $\ket{y_1,y_2}\ket{z}$ to $\ket{y_1,y_2}\ket{z+y_1+y_2}$. 
Notice that both $\unitaryU_{f_I,f_I}$ and $\unitaryU_{f_+}$ can be implemented by using controlled-NOT operators.

Now we present our protocol for the example of Fig.~\ref{fig:butterfly}. Let 
$$
\ket{\psi_S}_{(\sfS_1,\sfS_2)}=
\alpha_{00}\ket{0}_{\sfS_1}\ket{0}_{\sfS_2}+\alpha_{01}\ket{0}_{\sfS_1}\ket{1}_{\sfS_2}+
\alpha_{10}\ket{1}_{\sfS_1}\ket{0}_{\sfS_2}+\alpha_{11}\ket{1}_{\sfS_1}\ket{1}_{\sfS_2}
$$
be the state that has to be sent from the source nodes to the target nodes. Here $\sfS_1$ (resp.~$\sfS_2$)
is a register received by the source $s_1$ (resp.~$s_2$) through its virtual incoming edge.
\vspace{2mm}

\noindent{\bf Part 1.}
Node $s_1$ (resp.~$s_2$) introduces two registers~$\sfR_1$~and~$\sfR_2$ (resp.~$\sfR_3$~and~$\sfR_4$), 
and applies the operator $\unitaryU_{f_I,f_I}$ over the registers $(\sfS_1, \sfR_1,\sfR_2)$ (resp.~over the registers $(\sfS_2, \sfR_3,\sfR_4)$). The resulting state is
\[
\begin{split}
&
 \alpha_{00}\ket{\bmzero}_{(\sfS_1, \sfR_1, \sfR_2)} \ket{\bmzero}_{(\sfS_2, \sfR_3, \sfR_4)}\\
  +&
  \alpha_{01}\ket{\bmzero}_{(\sfS_1, \sfR_1, \sfR_2)}\ket{\bmone}_{(\sfS_2, \sfR_3, \sfR_4)}\\
  +
  &
  \alpha_{10}\ket{\bmone}_{(\sfS_1, \sfR_1, \sfR_2)}\ket{\bmzero}_{(\sfS_2, \sfR_3, \sfR_4)}\\
  +&
  \alpha_{11}\ket{\bmone}_{(\sfS_1, \sfR_1, \sfR_2)}\ket{\bmone}_{(\sfS_2, \sfR_3, \sfR_4)}.
\end{split}
\]
Hereafter, let $\bmzero$ and $\bmone$ denote strings of all-zero and all-one,
respectively, of appropriate length (three here). 
Then $\sfR_1$~and~$\sfR_2$ are sent to $n_3$ and $n_1$, respectively,
while $\sfR_3$~and~$\sfR_4$ are sent to $n_4$ and $n_1$, respectively. 

Then the protocol proceeds with the simulation of the coding operation 
performed at node~$n_1$ in the classical coding scheme of
Fig.~\ref{fig:classicalbutterfly}: node~$n_1$ prepares a new register~$\sfR_5$ 
and applies the operator $\unitaryU_{f_+}$ on registers $(\sfR_2,\sfR_4,\sfR_5)$. 
The resulting state is
\[
\begin{split}
&
 \alpha_{00}\ket{\bmzero}_{(\sfS_1,\sfR_1, \sfR_2)} \ket{\bmzero}_{(\sfS_2,\sfR_3, \sfR_4)} \ket{0}_{\sfR_5}\\
  +
  &
  \alpha_{01}\ket{\bmzero}_{(\sfS_1,\sfR_1, \sfR_2)} \ket{\bmone}_{(\sfS_2,\sfR_3, \sfR_4)}\ket{1}_{\sfR_5}\\
  +
 &
  \alpha_{10}\ket{\bmone}_{(\sfS_1,\sfR_1, \sfR_2)} \ket{\bmzero}_{(\sfS_2,\sfR_3, \sfR_4)}\ket{1}_{\sfR_5}\\
  +
  &\alpha_{11}\ket{\bmone}_{(\sfS_1,\sfR_1, \sfR_2)} \ket{\bmone}_{(\sfS_2,\sfR_3, \sfR_4)}\ket{0}_{\sfR_5}
.
 \end{split}
\]
The register~$\sfR_5$ is then sent to $n_2$.

The node~$n_2$ now prepares 
two registers~$\sfR_6$~and~$\sfR_7$, and applies the operator  $\unitaryU_{f_I,f_I}$
on the registers $(\sfR_5,\sf R_6,\sfR_7)$.
 The resulting state is
\[
\begin{split}
&
 \alpha_{00}\ket{\bmzero}_{(\sfS_1,\sfR_1, \sfR_2)} \ket{\bmzero}_{(\sfS_2,\sfR_3, \sfR_4)}  \ket{\bmzero}_{(\sfR_5,\sfR_6, \sfR_7)}\\
 +
&
 \alpha_{01}\ket{\bmzero}_{(\sfS_1,\sfR_1, \sfR_2)} \ket{\bmone}_{(\sfS_2,\sfR_3, \sfR_4)}\ket{\bmone}_{(\sfR_5,\sfR_6, \sfR_7)}\\\
 +
&
 \alpha_{10}\ket{\bmone}_{(\sfS_1,\sfR_1, \sfR_2)} \ket{\bmzero}_{(\sfS_2,\sfR_3, \sfR_4)}\ket{\bmone}_{(\sfR_5,\sfR_6, \sfR_7)}\\\
 + &
 \alpha_{11}\ket{\bmone}_{(\sfS_1,\sfR_1, \sfR_2)} \ket{\bmone}_{(\sfS_2,\sfR_3, \sfR_4)}\ket{\bmzero}_{(\sfR_5,\sfR_6, \sfR_7)}
 ,
 \end{split}
\]
and the registers~$\sfR_6$~and~$\sfR_7$ are sent to $n_3$ and $n_4$, respectively.

In the last step of the simulation, node $n_3$~(resp.~$n_4$) prepares one register~$\sfT_2$ (resp.~$\sfT_1$),
and applies the $\unitaryU_{f_+}$ to registers $(\sfR_1,\sfR_6,\sfT_2)$ (resp.~$(\sfR_3,\sfR_7,\sfT_1)$). 
The resulting state is
\[
\begin{split}
&
 \alpha_{00}\ket{\bmzero}_{(\sfS_1,\sfR_1, \sfR_2)} \ket{\bmzero}_{(\sfS_2,\sfR_3, \sfR_4)} \ket{\bmzero}_{(\sfR_5,\sfR_6, \sfR_7)}\ket{0}_{\sfT_1} \ket{0}_{\sfT_2}\\
 +
&
 \alpha_{01}\ket{\bmzero}_{(\sfS_1,\sfR_1, \sfR_2)} \ket{\bmone}_{(\sfS_2,\sfR_3, \sfR_4)}\ket{\bmone}_{(\sfR_5,\sfR_6, \sfR_7)}\ket{0}_{\sfT_1} \ket{1}_{\sfT_2}\\\
 +
&
 \alpha_{10}\ket{\bmone}_{(\sfS_1,\sfR_1, \sfR_2)} \ket{\bmzero}_{(\sfS_2,\sfR_3, \sfR_4)}\ket{\bmone}_{(\sfR_5,\sfR_6, \sfR_7)} \ket{1}_{\sfT_1} \ket{0}_{\sfT_2}\\\
 + &
 \alpha_{11}\ket{\bmone}_{(\sfS_1,\sfR_1, \sfR_2)} \ket{\bmone}_{(\sfS_2,\sfR_3, \sfR_4)}\ket{\bmzero}_{(\sfR_5,\sfR_6, \sfR_7)}\ket{1}_{\sfT_1} \ket{1}_{\sfT_2},
 \end{split}
\]
and registers $\sfT_1$ and $\sfT_2$ are sent to nodes $t_1$ and $t_2$, respectively.

At this point of our protocol, node $s_1$ owns register $\sfS_1$, node $s_2$ owns register $\sfS_2$, node $n_1$ owns registers $\sfR_2$ and $\sfR_4$, node $n_2$ owns 
register $\sfR_5$, node $n_3$ owns register $\sfR_1$ and $\sfR_6$, node $n_4$ owns registers $\sfR_3$ and $\sfR_7$, node $t_1$ owns register $\sfT_1$, and node $t_2$
owns register $\sfT_2$.

\vspace{2mm}
\noindent{\bf Part 2.}
Let us take a reverse topological order of the internal nodes of the graph, for example $(n_3,n_4,n_2,n_1)$.

Node $n_3$ first measures its registers $\sfR_1$ and $\sfR_6$ in the Hadamard basis 
(i.e., the Fourier basis for $|\Sigma|=2$), 
giving outcomes $a_1$ and $a_2$ respectively. 
The resulting state is
\[
\begin{split}
&
  \alpha_{00}\ket{0}_{\sfS_1} \ket{a_1}_{\sfR_1} \ket{0}_{\sfR_2}  \ket{\bmzero}_{(\sfS_2,\sfR_3, \sfR_4)}\ket{0}_{\sfR_5} \ket{a_2}_{\sfR_6}\ket{0}_{\sfR_7}\ket{0}_{\sfT_1} \ket{0}_{\sfT_2}\\
  +
  (-1)^{a_2}
 &
  \alpha_{01}\ket{0}_{\sfS_1}\ket{a_1}_{\sfR_1} \ket{0}_{\sfR_2}  \ket{\bmone}_{(\sfS_2,\sfR_3, \sfR_4)}\ket{1}_{\sfR_5} \ket{a_2}_{\sfR_6}\ket{1}_{\sfR_7}\ket{0}_{\sfT_1} \ket{1}_{\sfT_2}\\
  +
   (-1)^{a_1+a_2}
 &
  \alpha_{10}\ket{1}_{\sfS_1}\ket{a_1}_{\sfR_1} \ket{1}_{\sfR_2}  \ket{\bmzero}_{(\sfS_2,\sfR_3, \sfR_4)}\ket{1}_{\sfR_5} \ket{a_2}_{\sfR_6}\ket{1}_{\sfR_7}\ket{1}_{\sfT_1} \ket{0}_{\sfT_2}\\
  + 
  (-1)^{a_1}
  &
  \alpha_{11}\ket{1}_{\sfS_1}\ket{a_1}_{\sfR_1} \ket{1}_{\sfR_2} \ket{\bmone}_{(\sfS_2,\sfR_3, \sfR_4)} \ket{0}_{\sfR_5} \ket{a_2}_{\sfR_6}\ket{0}_{\sfR_7}\ket{1}_{\sfT_1} \ket{1}_{\sfT_2},\\
  \end{split}
\]
and $\sfR_1$ and $\sfR_6$ can be disregarded. Node $n_3$ then sends the bit $a_1$ to $s_1$, 
and the bit $a_2$ to $n_2$, using classical communication along the edges $(s_1,n_3)$ and $(n_2,n_3)$ but in the reverse direction. 
Node $s_1$ then applies on its register $\sfS_1$ the quantum operator mapping, for each $x\in\bbF_2$, the state $\ket{x}$ to $(-1)^{-a_1x}\ket{x}$.
Node $n_2$ applies on its register $\sfR_5$ the quantum operator mapping, for each $x\in\bbF_2$, the state $\ket{x}$ to $(-1)^{-a_2x}\ket{x}$. The quantum state becomes 
\[
\begin{split}
&
  \alpha_{00}\ket{0}_{\sfS_1} \ket{0}_{\sfR_2}  \ket{\bmzero}_{(\sfS_2,\sfR_3, \sfR_4)}\ket{0}_{\sfR_5}\ket{0}_{\sfR_7}\ket{0}_{\sfT_1} \ket{0}_{\sfT_2}\\
  +
 &
  \alpha_{01}\ket{0}_{\sfS_1}\ket{0}_{\sfR_2}  \ket{\bmone}_{(\sfS_2,\sfR_3, \sfR_4)}\ket{1}_{\sfR_5} \ket{1}_{\sfR_7}\ket{0}_{\sfT_1} \ket{1}_{\sfT_2}\\
  +
 &
  \alpha_{10}\ket{1}_{\sfS_1}\ket{1}_{\sfR_2}  \ket{\bmzero}_{(\sfS_2,\sfR_3, \sfR_4)}\ket{1}_{\sfR_5} \ket{1}_{\sfR_7}\ket{1}_{\sfT_1} \ket{0}_{\sfT_2}\\
  + 
  &
  \alpha_{11}\ket{1}_{\sfS_1} \ket{1}_{\sfR_2} \ket{\bmone}_{(\sfS_2,\sfR_3, \sfR_4)} \ket{0}_{\sfR_5} \ket{0}_{\sfR_7}\ket{1}_{\sfT_1} \ket{1}_{\sfT_2},\\
  \end{split}
\]
where $\sfR_1$ and $\sfR_6$ have been disregarded.

Similarly, node $n_4$ measures its registers $\sfR_3$ and $\sfR_7$ in the Hadamard basis, giving outcomes $b_1$ and $b_2$ respectively. 
The resulting state is
\[
\begin{split}
&
  \alpha_{00}\ket{0}_{\sfS_1} \ket{0}_{\sfR_2}  \ket{0}_{\sfS_2}\ket{b_1}_{\sfR_3}\ket{0}_{\sfR_4}\ket{0}_{\sfR_5}\ket{b_2}_{\sfR_7}\ket{0}_{\sfT_1} \ket{0}_{\sfT_2}\\
  +
  (-1)^{b_1+b_2}
 &
\alpha_{01}\ket{0}_{\sfS_1} \ket{0}_{\sfR_2}  \ket{1}_{\sfS_2}\ket{b_1}_{\sfR_3}\ket{1}_{\sfR_4}\ket{1}_{\sfR_5}\ket{b_2}_{\sfR_7}\ket{0}_{\sfT_1} \ket{1}_{\sfT_2}\\
  +
   (-1)^{b_2}
 &
\alpha_{10}\ket{1}_{\sfS_1} \ket{1}_{\sfR_2}  \ket{0}_{\sfS_2}\ket{b_1}_{\sfR_3}\ket{0}_{\sfR_4}\ket{1}_{\sfR_5}\ket{b_2}_{\sfR_7}\ket{1}_{\sfT_1} \ket{0}_{\sfT_2}\\
  + 
   (-1)^{b_1}
  &
\alpha_{11}\ket{1}_{\sfS_1} \ket{1}_{\sfR_2}  \ket{1}_{\sfS_2}\ket{b_1}_{\sfR_3}\ket{1}_{\sfR_4}\ket{0}_{\sfR_5}\ket{b_2}_{\sfR_7}\ket{1}_{\sfT_1} \ket{1}_{\sfT_2}.
  \end{split}
\]
Node $n_4$ then sends $b_1$ to node $s_2$, and $b_2$ to node $n_2$  using classical communication.
Node $s_1$ then applies on its register $\sfS_2$ the quantum operator mapping, for each $x\in\bbF_2$, the state $\ket{x}$ to $(-1)^{-b_1x}\ket{x}$.
Node $n_2$ applies on its register $\sfR_5$ the quantum operator mapping, for each $x\in\bbF_2$, the state $\ket{x}$ to $(-1)^{-b_2x}\ket{x}$. The 
quantum state becomes 
\[
\begin{split}
&
  \alpha_{00}\ket{0}_{\sfS_1} \ket{0}_{\sfR_2}  \ket{0}_{\sfS_2}\ket{0}_{\sfR_4}\ket{0}_{\sfR_5}\ket{0}_{\sfT_1} \ket{0}_{\sfT_2}\\
  +
 &
\alpha_{01}\ket{0}_{\sfS_1} \ket{0}_{\sfR_2}  \ket{1}_{\sfS_2}\ket{1}_{\sfR_4}\ket{1}_{\sfR_5}\ket{0}_{\sfT_1} \ket{1}_{\sfT_2}\\
  +
 &
\alpha_{10}\ket{1}_{\sfS_1} \ket{1}_{\sfR_2}  \ket{0}_{\sfS_2}\ket{0}_{\sfR_4}\ket{1}_{\sfR_5}\ket{1}_{\sfT_1} \ket{0}_{\sfT_2}\\
  + 
  &
\alpha_{11}\ket{1}_{\sfS_1} \ket{1}_{\sfR_2}  \ket{1}_{\sfS_2}\ket{1}_{\sfR_4}\ket{0}_{\sfR_5}\ket{1}_{\sfT_1} \ket{1}_{\sfT_2},
  \end{split}
\]
where $\sfR_3$ and $\sfR_7$ have been disregarded.

Node $n_2$ then measures its registers $\sfR_5$ in the Hadamard basis, giving outcomes $c$. The state becomes
\[
\begin{split}
&
  \alpha_{00}\ket{0}_{\sfS_1} \ket{0}_{\sfR_2}  \ket{0}_{\sfS_2}\ket{0}_{\sfR_4}\ket{c}_{\sfR_5}\ket{0}_{\sfT_1} \ket{0}_{\sfT_2}\\
  +
  (-1)^{c}
 &
\alpha_{01}\ket{0}_{\sfS_1} \ket{0}_{\sfR_2}  \ket{1}_{\sfS_2}\ket{1}_{\sfR_4}\ket{c}_{\sfR_5}\ket{0}_{\sfT_1} \ket{1}_{\sfT_2}\\
  +
  (-1)^{c}
 &
\alpha_{10}\ket{1}_{\sfS_1} \ket{1}_{\sfR_2}  \ket{0}_{\sfS_2}\ket{0}_{\sfR_4}\ket{c}_{\sfR_5}\ket{1}_{\sfT_1} \ket{0}_{\sfT_2}\\
  + 
  &
\alpha_{11}\ket{1}_{\sfS_1} \ket{1}_{\sfR_2}  \ket{1}_{\sfS_2}\ket{1}_{\sfR_4}\ket{c}_{\sfR_5}\ket{1}_{\sfT_1} \ket{1}_{\sfT_2}.
  \end{split}
\]
The value $c$ is then sent to node $n_1$.
The register $\sfR_5$ can be disregarded, and the phase introduced is corrected in the following way: node $n_1$ applies
on its registers $(\sfR_2,\sfR_4)$ the unitary operator mapping, for any $x,y\in\bbF_2$, the state $\ket{x,y}$ to 
the state $(-1)^{-cf_+(x,y)}\ket{x,y}=(-1)^{-c(x+y)}\ket{x,y}$.
The resulting state is 
\[
\begin{split}
&
  \alpha_{00}\ket{0}_{\sfS_1} \ket{0}_{\sfR_2}  \ket{0}_{\sfS_2}\ket{0}_{\sfR_4}\ket{0}_{\sfT_1} \ket{0}_{\sfT_2}\\
  +
 &
\alpha_{01}\ket{0}_{\sfS_1} \ket{0}_{\sfR_2}  \ket{1}_{\sfS_2}\ket{1}_{\sfR_4}\ket{0}_{\sfT_1} \ket{1}_{\sfT_2}\\
  +
 &
\alpha_{10}\ket{1}_{\sfS_1} \ket{1}_{\sfR_2}  \ket{0}_{\sfS_2}\ket{0}_{\sfR_4}\ket{1}_{\sfT_1} \ket{0}_{\sfT_2}\\
  + 
  &
\alpha_{11}\ket{1}_{\sfS_1} \ket{1}_{\sfR_2}  \ket{1}_{\sfS_2}\ket{1}_{\sfR_4}\ket{1}_{\sfT_1} \ket{1}_{\sfT_2}.
  \end{split}
\]
Notice that the same procedure to remove register $\sfR_5$ would have worked even if $f_+$ is not a linear function.
This is the crucial observation that enables us to simulate nonlinear classical protocols as well.

Finally, node $n_1$ measures its registers $\sfR_2$ and $\sfR_4$ in the  Hadamard basis, giving outcomes $d_1$ and
$d_2$ respectively. It sends $d_1$ to node $s_1$, and $d_2$ to node $s_2$, respectively. 
The registers $\sfR_2$ and $\sfR_4$ can then be disregarded, 
and the phase introduced is corrected in the following way: node $s_1$ (resp.~$s_2$) applies
on its register $\sfS_1$ (resp.~$\sfS_2$) the unitary operator mapping, for any $x\in\bbF_2$, 
the state $\ket{x}$ to the state $(-1)^{-d_1x}\ket{x}$ (resp.~to $(-1)^{-d_2x}\ket{x}$). The quantum state becomes
\[
\begin{split}
&
  \alpha_{00}\ket{0}_{\sfS_1}   \ket{0}_{\sfS_2}\ket{0}_{\sfT_1} \ket{0}_{\sfT_2}\\
  +
 &
\alpha_{01}\ket{0}_{\sfS_1}  \ket{1}_{\sfS_2}\ket{0}_{\sfT_1} \ket{1}_{\sfT_2}\\
  +
 &
\alpha_{10}\ket{1}_{\sfS_1}  \ket{0}_{\sfS_2}\ket{1}_{\sfT_1} \ket{0}_{\sfT_2}\\
  + 
  &
\alpha_{11}\ket{1}_{\sfS_1}  \ket{1}_{\sfS_2}\ket{1}_{\sfT_1} \ket{1}_{\sfT_2}.
  \end{split}
\]

\vspace{3mm}
\noindent{\bf Part 3.}
Node $s_1$ (resp.~$s_2$) measures its quantum register $\sfS_1$ (resp.~$\sfS_2$) in the 
Hadamard basis, giving outcomes $e_1$ (resp.~$e_2$). The state becomes
\[
\begin{split}
&
  \alpha_{00}\ket{e_1}_{\sfS_1}   \ket{e_2}_{\sfS_2}\ket{0}_{\sfT_1} \ket{0}_{\sfT_2}\\
  +
  (-1)^{e_2}
 &
\alpha_{01}\ket{e_1}_{\sfS_1}  \ket{e_2}_{\sfS_2}\ket{0}_{\sfT_1} \ket{1}_{\sfT_2}\\
  +
  (-1)^{e_1}
 &
\alpha_{10}\ket{e_1}_{\sfS_1}  \ket{e_2}_{\sfS_2}\ket{1}_{\sfT_1} \ket{0}_{\sfT_2}\\
  + 
  (-1)^{e_1+e_2}
  &
\alpha_{11}\ket{e_1}_{\sfS_1}  \ket{e_2}_{\sfS_2}\ket{1}_{\sfT_1} \ket{1}_{\sfT_2}.
  \end{split}
\]
The registers $\sfS_1$ and $\sfS_2$ can then be disregarded. Then $e_1$ is sent 
from the sources $s_1$ to the target $t_1$ using classical communication and, 
similarly, $e_2$ is sent from $s_2$ to $t_2$. This is done using one bit of communication per edge by 
using the original coding protocol for the butterfly graph:  the bit $e_1$ is sent through edges
$(s_1,n_1)$ and $(s_1,n_3)$, the bit $e_2$ is sent through edges $(s_2,n_1)$ and $(s_2,n_4)$,
and the bit $e_1+e_2$ is sent through edges $(n_1,n_2)$, $(n_2,n_3)$ and $(n_2,n_4)$. The bit
$e_1$ can then be recovered at node $t_1$ and the bit $e_2$ can be recovered at node $t_2$.

Node $t_1$ (resp.~$t_2$) finally applies on its register $\sfT_1$ (resp.~$\sfT_2$) 
the unitary operator mapping, for any $x\in\bbF_2$, the state $\ket{x}$ to 
the state $(-1)^{-e_1x}\ket{x}$ (resp.~to $(-1)^{-e_2x}\ket{x}$). Now the quantum state becomes
the desired state
$$
\alpha_{00}\ket{0}_{\sfT_1}\ket{0}_{\sfT_2}+\alpha_{01}\ket{0}_{\sfT_1}\ket{1}_{\sfT_2}+
\alpha_{10}\ket{1}_{\sfT_1}\ket{0}_{\sfT_2}+\alpha_{11}\ket{1}_{\sfT_1}\ket{1}_{\sfT_2}
=\ket{\psi_S}_{(\sfT_1,\sfT_2)}.
$$
This concludes our protocol.


\end{document}